\documentclass[reqno,a4paper]{amsart}

\usepackage{amsmath}%
\usepackage{txfonts}
\usepackage{amsfonts}%
\usepackage{amssymb}%
\usepackage{graphicx}
\theoremstyle{plain}
\newtheorem{theorem}{Theorem}[section]
\newtheorem{lemma}[theorem]{Lemma}

\newtheorem{proposition}[theorem]{Proposition}
\newtheorem{conjecture}[theorem]{Conjecture}

\newtheorem*{claim*}{Claim}

\theoremstyle{definition}
\newtheorem{definition}[theorem]{Definition}

\newtheorem{example}[theorem]{Example}
\newtheorem{remark}[theorem]{Remark}

\DeclareMathOperator*{\argmax}{arg\,max}

\DeclareMathOperator*{\inter}{int}

\newcommand{\R}{\ensuremath{\mathbb{R}}}

\newcommand\BR{{\mathcal B \! \mathcal R \!}}
\newcommand\bp{{\bar p}}
\newcommand\bq{{\bar q}}

\providecommand{\q}[1]{\lq#1\/\rq}

\newcommand\po[1]{{u^{#1}}}
\newcommand\poa{{\po{A}}}
\newcommand\pob{{\po{B}}}

\newcommand\hpo[1]{{{\hat u}^{#1}}}
\newcommand\hpoa{{\hpo{A}}}
\newcommand\hpob{{\hpo{B}}}

\newcommand\Pure[1]{{S^{#1}}}

\newcommand\Sig[1]{{\Sigma^{#1}}}
\newcommand\Siga{{\Sig{A}}}
\newcommand\Sigb{{\Sig{B}}}

\newcommand\BRa{{\BR_A}}
\newcommand\BRb{{\BR_B}}

\newcommand\Regx[2]{{R^{#1}_{#2}}}
\newcommand\Rega[1]{{\Regx{A}{#1}}}
\newcommand\Regb[1]{{\Regx{B}{#1}}}

\newcommand\indiffx[3]{{Z^{#1}_{{#2}{#3}}}}
\newcommand\indiffa[2]{{\indiffx{A}{#1}{#2}}}
\newcommand\indiffb[2]{{\indiffx{B}{#1}{#2}}}

\newcommand\nasha{{E^A}}
\newcommand\nashb{{E^B}}
\newcommand\nash{{(\nasha,\nashb)}}

\begin{document}
\title[Payoff Performance of Fictitious Play]{Payoff Performance of Fictitious Play}
\author[Georg Ostrovski]{Georg Ostrovski}
\address[Georg Ostrovski]{Mathematics Institute,
Zeeman Building, University of Warwick,
Coventry CV4 7AL}
\email{g.ostrovski@gmail.com}

\author[Sebastian van Strien]{Sebastian van Strien}
\address[Sebastian van Strien]{Department of Mathematics,
Imperial College London, London SW7 2AZ}
\email{s.van-strien@imperial.ac.uk}

\thanks{The authors would like to thank Christopher Harris, Sergiu Hart and Ulrich Berger
for useful discussions and comments.}

\begin{abstract}
We investigate how well continuous-time fictitious play in two-player games performs in terms of 
average payoff, particularly compared to Nash equilibrium payoff. We show 
that in many games, fictitious play outperforms Nash equilibrium on average or 
even at all times, and moreover that any game is linearly equivalent 
to one in which this is the case. Conversely, we provide conditions 
under which Nash equilibrium payoff dominates fictitious play payoff.
A key step in our analysis is to show that
fictitious play dynamics asymptotically converges the set of coarse correlated equilibria 
(a fact which is implicit in the literature).
\end{abstract}

\maketitle

Continuous-time fictitious play (FP) has been first introduced by Brown~\cite{Brown1949,Brown1951}
and it has since been a standard model for myopic learning, often used 
as a convenient reference algorithm due to its computational simplicity 
(see, for example,~\cite{Fudenberg1998,Young2005b}).
It has been shown to converge to Nash equilibrium in many important classes of games, such
as zero-sum games~\cite{Hofbauer1995}, non-degenerate $2 \times n$ games~\cite{Berger2005}, 
non-degenerate quasi-supermodular games with diminishing returns or of dimension 
$3 \times n$ or $4 \times 4$~\cite{Berger2008,Berger2007}, and others. 
On the other hand, convergence to Nash equilibrium (even when it is unique)
is not guaranteed, as demonstrated by Shapley's famous example~\cite{Shapley1964} of
a $3 \times 3$ Rock-Paper-Scissors-like game with a stable limit cycle for FP. 
Note that in  Rock-Paper-Scissors-like games with an attracting 
limit cycle, the limit cycle is generally not globally attracting: uncountably many orbits are still attracted to the 
Nash equilibrium, see \cite[Theorem 1.1]{VanStrien2010}. 

The question therefore arises whether in the non-convergent case the payoff
to the players along trajectories of FP compares favourably to 
Nash equilibrium payoff. In this paper we investigate the relation between
Nash equilibrium payoff and average payoff along FP trajectories. 
In particular, we show that in many two-player games, FP
may in the long run earn a higher payoff to both players than Nash equilibrium play, 
either on average, or even at all times. We also show that every two-player game
is \q{linearly equivalent} to one in which FP Pareto dominates
Nash equilibrium (at all times, along every non-equilibrium FP orbit). 
Conversely, we give conditions under which FP is dominated by Nash equilibrium
in terms of payoff, and show numerical examples for this (rather atypical) behaviour. 

The paper is organized as follows. 
In Section~\ref{sec:fp} we introduce basic notation. 
In Section~\ref{sec:limit_set} we analyse the limiting behaviour of 
FP dynamics and show that FP converges
to the so-called set of coarse correlated equilibria.
In Section~\ref{sec:fp_vs_ne} we use this
to compare the payoff along the limit sets with the Nash equilibrium payoffs. 
Ultimately, this allows us to show that every 
bimatrix game is linearly equivalent to one in which 
FP Pareto dominates Nash equilibrium and we discuss 
the conditions governing the payoff comparison of these two.
In Section~\ref{sec:ex_family} we present a particular family
of $3 \times 3$ games in which
FP yields higher average payoff
to both players than Nash equilibrium.
In Section~\ref{sec:fp_worse} we investigate the possibility of 
games in which Nash equilibrium play dominates FP. We also 
deduce conditions for this and numerically determine examples in which this
is the case. The discussion shows that these examples are relatively \q{rare}. 
Finally, in Section~\ref{sec:conclusion_fp_performance} we discuss
the implications of these results for the notions of 
equilibrium (in the context of payoff performance of learning algorithms)
and game equivalence.

\section{Notation and standard facts}
\label{sec:fp}

For $A = (a_{ij}),B=(b_{ij}) \in \R^{m \times n}$, we denote by $(A,B)$ 
a \emph{bimatrix game} with players A and B
having \emph{pure strategies} $\Pure{A} = \{1,\ldots,m\}$ and $\Pure{B} = \{1,\ldots,n\}$. 
We call $S = \Pure{A} \times \Pure{B}$ the \emph{joint strategy space}, 
and we call a probability distribution over $S$ a \emph{joint probability distribution}.
By $\Siga \subset \R^{1\times m}$ and $\Sigb \subset \R^{n \times 1}$ 
we denote the $(m-1)$- and $(n-1)$-dimensional simplices of \emph{mixed
strategies} of the two players, and we implicitly identify the pure strategy
$i \in \Pure{A}$ with the $i$th unit vector in $\R^{1 \times m}$ and $j \in \Pure{B}$
with the $j$th unit vector in $\R^{n \times 1}$. We write $\Sigma = \Siga \times \Sigb$ for the
space of \emph{mixed strategy profiles}. Note that this can be seen as a proper subset
of the set of joint probability distributions.

The \emph{payoffs} to players A and B from playing the 
pure strategy profile $(i,j) \in \Pure{A} \times \Pure{B}$ are $a_{ij}$ and $b_{ij}$, respectively. 
By linearity, their expected payoffs from playing a mixed strategy 
profile $(x,y) \in \Sigma = \Siga \times \Sigb$ are 
\[
\poa(x,y) = x A y \quad \text{ and } \quad \pob(x,y) = x B y.
\]

The players' \emph{best response correspondences} $\BRa \colon \Sigb \to \Siga$ 
and $\BRb \colon \Siga \to \Sigb$ are given by 
\[
\BRa(q) \coloneqq \argmax_{\bp \in \Siga}  \bp A q \quad 
\text{ and } \quad \BRb(p) \coloneqq \argmax_{\bq \in \Sigb} p B \bq.
\]
We further denote the maximal-payoff functions 
\[
\bar A(q) \coloneqq \max_{ \bp \in \Siga} \bp A q 
\quad \text{ and } \quad \bar B(p) \coloneqq \max_{\bq \in \Sigb} p B \bq, 
\]
so that $\bar A(q) = \poa(\bp,q)$ 
for $\bp \in \BRa(q)$ and $\bar B(p) = \pob(p,\bq)$ 
for $\bq \in \BRb(p)$. Observe that
$\bar A(q) = \max_i\,(A q)_i$  and $\bar B(p) = \max_j\,(p B)_j$: 
the maximal payoff to player A given player B's strategy $q$ is 
equal to the maximal entry of the vector $A q$, and similarly for player B.

For generic bimatrix games, the best response correspondences
$\BRa \colon \Sigb \to \Siga$ and $\BRb \colon \Siga \to \Sigb$ 
are almost everywhere single-valued, with the exception of a finite number of hyperplanes. 
The singleton value taken by $\BRa$ whenever it is single-valued is always a pure strategy of player A.
When $\BRa(p)$ is not a singleton, it is the set of convex combinations of a subset of 
$\{ e_i \colon i \in \Pure{A}\}$, that is, a face of the simplex $\Siga$, or possibly all of $\Siga$.
The analogous statement holds for $\BRb$. 

It follows that $\Siga$ and $\Sigb$ can be divided into respectively $n$ and $m$ 
regions (in fact, closed convex polytopes):
\begin{equation*}
\begin{split}
\Regb{j} &\coloneqq \BRb^{-1}(j) \subseteq \Siga \qquad \text{for } j \in \Pure{B}, \\
\Rega{i} &\coloneqq \BRa^{-1}(i) \subseteq \Sigb \qquad \text{for } i \in \Pure{A}.  
\end{split}
\end{equation*}
We will call $\Rega{i}$ the \emph{preference region} of strategy $i$ for player A, as
it is the (closed) subset of the second player's strategies against which player A expects the highest
payoff by playing strategy $i$; similarly, for $\Regb{j}$. 

For a generic game $(A,B)$, the subset of $\Sigb$ on which $\BRa$ 
contains two distinct pure strategies $i,i' \in \Pure{A}$ 
(and hence all their convex combinations)
is contained in a codimension-one hyperplane of $\Sigb$:
\[
\indiffa{i}{i'} \coloneqq \{ q \in \Sigb : (A q)_i = (A q)_{i'} \geq (A q)_k ~\forall k \in \Pure{A}\} 
= \Rega{i} \cap \Rega{i'} \subseteq \Sigb.
\]
Analogously, for $j,j' \in \Pure{B}$,
\[
\indiffb{j}{j'} \coloneqq \{ p \in \Siga : (p B)_j = (p B)_{j'} \geq (p B)_l ~\forall l \in \Pure{B}\} 
= \Regb{j} \cap \Regb{j'} \subseteq \Siga.
\]
These hyperplanes are subsets of linear codimension-one subspaces of $\Sigb$ and $\Siga$, respectively.
See Figure~\ref{fig:game_geom} for an illustration in the case $n = m = 3$. 
We call these sets the \emph{indifference sets} of players A and B. 

\begin{figure}
\centering
\includegraphics[width=\textwidth,trim=60 0 20 0]{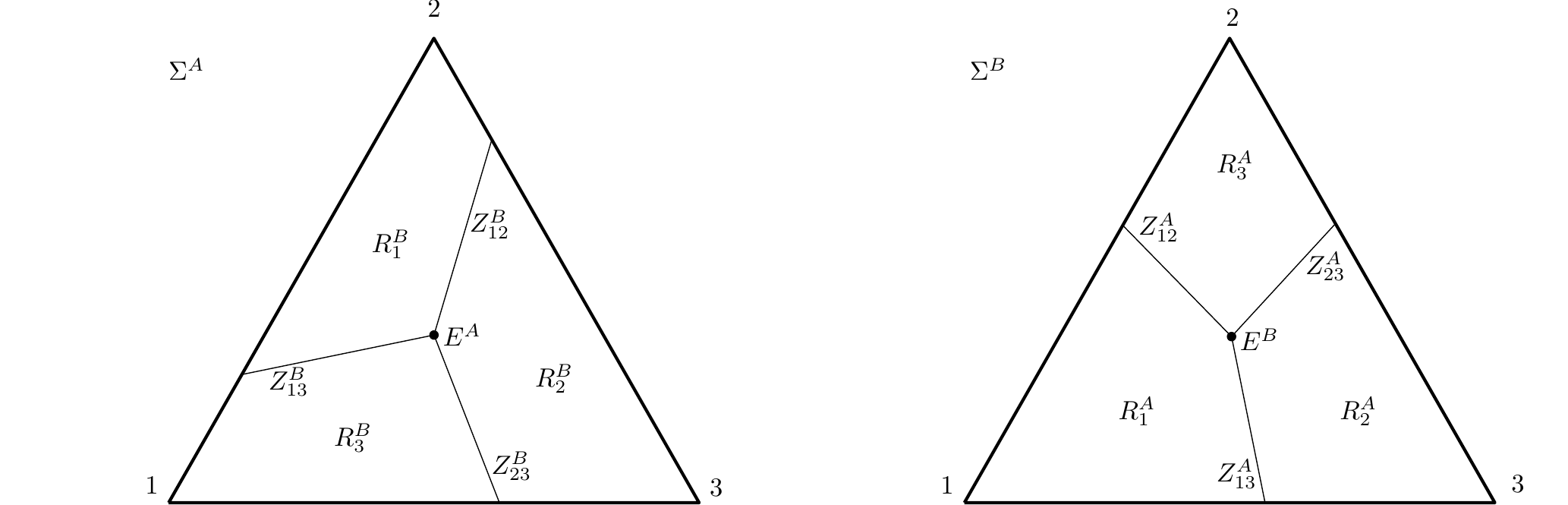}
\caption{Geometry of a $3 \times 3$ bimatrix game. The spaces of mixed strategies $\Siga$ and $\Sigb$ 
are each a simplex spanned by three vertices (the pure strategies). Note the closed convex 
preference regions $\Regb{j} \subset \Siga$ and $\Rega{i} \subset \Sigb$, their intersections as 
indifference sets $\indiffb{j}{j'}$ and $\indiffa{i}{i'}$, and the projections to $\Siga$ and $\Sigb$
of the (in this case, unique) Nash equilibrium $\nash$ at the intersection of all these sets.}
\label{fig:game_geom}
\end{figure}

\begin{definition}
A mixed strategy profile $(\bp, \bq) \in \Sigma$ is a \emph{Nash equilibrium}, if
$\bp \in \BRa(\bq)$ and $\bq \in \BRb(\bp)$. 
If a Nash equilibrium lies in the interior of $\Sigma$, it is called \emph{completely mixed}. 
\end{definition}

The following lemma is a standard fact and easy to check.

\begin{lemma}\label{lem:ne_int}
 The point $\nash \in \inter (\Sigma)$ is a 
 (completely mixed) Nash equilibrium of an $m \times n$ bimatrix 
 game $(A,B)$ if and only if, for all $i, i' = 1, \ldots, m$ and $j,j' = 1, \ldots, n$,
\[
 (A \nashb)_i = (A \nashb)_{i'} \quad \text{ and } \quad (\nasha B)_j = (\nasha B)_{j'}.
\]
\end{lemma}
Note that this implies that $\nasha \in \Regb{j}$ and $\nashb \in \Rega{i}$, for all $i,j$.

From the various ways to define continuous-time FP, we follow the 
approach taken in~\cite{Harris1998}.
We define a \emph{continuous-time fictitious play process} 
$(p(t),q(t)) \in \Sigma$, $t \geq 1$, 
as a solution to the differential inclusion 
\begin{equation}\label{eq:fp}
\dot p (t) \in \frac{1}{t} (\BRa(q(t)) - p(t)), \quad 
\dot q (t) \in \frac{1}{t} (\BRa(p(t)) - q(t)),
\end{equation}
with some initial condition $(p(1),q(1)) \in \Sigma$ (see, for example,~\cite{Harris1998,Hofbauer1995}).

Alternatively, as in~\cite{Harris1998}, we can denote 
by $x (t)$ and $y (t)$ the strategies played by the two players at time $t \geq 0$, 
where $x \colon [0,\infty) \rightarrow \Siga$ and 
$y \colon [0,\infty) \rightarrow \Sigb$ are assumed to be measurable functions. 
We write the average \emph{(empirical) past play} 
of the respective players from time $0$ through $t$ as
\[
p(t) \coloneqq \frac{1}{t}\int_0^t x(s)\, ds \quad 
\text{ and } \quad  q(t) \coloneqq \frac{1}{t}\int_0^t y(s)\, ds. 
\]
  Then continuous-time FP is 
  given by the rule expressed in the following integral inclusions:
\[
  x(t) \in \BRa( q(t) )\quad \text{ and } \quad  y(t) \in \BRb( p(t) ) \qquad \text{for } t \geq 1
\]
  and $(x(t),y(t)) \in \Sigma$ arbitrary for $0 \leq t < 1$. 
  Defined this way, $(p(t),q(t))$, $t \geq 1$, is a solution of the differential inclusion~\eqref{eq:fp}
  with initial condition $p(1) = \int_0^1 x(s)\,ds$ and $q(1) = \int_0^1 y(s)\,ds$.

\begin{definition} \label{def:equivalence}
We say that two $m \times n$ bimatrix games $(A,B)$ and $(\tilde A,\tilde B)$ 
are \emph{(linearly) equivalent}, $(A,B) \sim (\tilde A,\tilde B)$, if 
the matrix $\tilde{A}$ can be obtained by multiplying $A$
with a positive constant $c > 0$ and adding constants $c_1,\ldots,c_n \in \R$ to the matrix columns, 
and $\tilde B$ can be obtained from $B$ by multiplication with $d > 0$ 
and addition of $d_1, \ldots,d_m \in \R$ to its rows:
\[
{\tilde a}_{ij} = c \cdot a_{ij} + c_j \quad \text{ and } \quad {\tilde b}_{ij} = d \cdot b_{ij} + d_i
\qquad \text{for } i = 1,\ldots,m \text{ and } j = 1,\ldots,n.
\]
\end{definition}

The following lemma follows by direct computation. 

\begin{lemma}\label{lem:lin_equiv_br_equiv}
Let $(A,B)$ and $(\tilde A,\tilde B)$ be two $m \times n$ bimatrix games. 
If $(A,B)$ and $(\tilde A,\tilde B)$ are linearly equivalent, 
then their best response correspondences coincide: $\BRa \equiv \BR_{\tilde A}$
and $\BRb \equiv \BR_{\tilde B}$.  
In particular, they have the same Nash equilibria, the same preference regions and 
indifference sets, and give rise to the same FP dynamics. 
\end{lemma}

We call two bimatrix games giving rise to the same best response correspondences \emph{dynamically equivalent}.

\section{Limit set for FP} 
\label{sec:limit_set}

In this section we study the long-term behaviour of (continuous-time) FP. 
It has been known since Shapley's famous version of 
the Rock-Paper-Scissors game~\cite{Shapley1964}
that FP does not necessarily converge 
to a Nash equilibrium even when the latter is unique, 
and can converge to a limit cycle instead.
In fact, convergence to a unique Nash equilibrium in the interior of $\Sigma$ 
seems to be rather the exception than the rule: 
It is a standing conjecture that such Nash equilibrium can only 
be stable for FP dynamics,
if the game is equivalent to a zero-sum game \cite{Hofbauer1995}. 
As an aside, we remark that applying a {\lq}spherical{\rq} projection of 
the dynamics of a zero-sum games with a unique interior Nash Equilibrium 
(projecting from the Nash Equilibrium onto the boundary of the simplex), gives rise to 
Hamiltonian dynamics, see \cite{VanStrien2011}.

We will show that every FP orbit converges to 
(a subset of) the set of so-called \q{coarse correlated equilibria}, 
sometimes also referred to as the \q{Hannan set}
(see~\cite{Hannan1957,Young2005b,Hart2005}). 
In fact, this result follows directly 
from the \q{belief affirming} property of FP\footnote{The authors thank
Sergiu Hart for pointing out this connection between Theorems~\ref{thm:lim_set}
and~\ref{thm:lim_av_payoff} and~\cite{Monderer1997,Fudenberg1995} when shown an early 
draft of this paper.}, 
shown in~\cite{Monderer1997}. 
However, to the best of our knowledge, the conclusion 
that FP has its limit set contained in the set of coarse correlated equilibria
has not been mentioned in the literature. 
We also provide a slightly different proof of this fact.

The following definition can be found in~\cite{Moulin1978}.

\begin{definition}\label{def:cce}
A joint probability distribution $P = (p_{ij})$ over $S$ 
is a \emph{coarse correlated equilibrium (CCE)} for the bimatrix game $(A,B)$ if 
\[
\sum_{i,j} a_{i'j} p_{ij} \leq \sum_{i,j} a_{ij} p_{ij}
\]
and
\[
\sum_{i,j} b_{ij'} p_{ij} \leq \sum_{i,j} b_{ij} p_{ij}
\]
for all $(i',j') \in S$.
The set of CCE is also called the \emph{Hannan set}.
\end{definition}

One way of viewing the concept of CCE is in terms of 
the notion of \emph{regret}. Let us assume that two players are (repeatedly or continuously)
playing a bimatrix game $(A,B)$, and let $P(t) = (p_{ij}(t))$ 
be the empirical joint distribution of their past play through time $t$, that is, 
$p_{ij}(t)$ represents the fraction of time of the strategy profile $(i,j)$
along their play through time $t$.
For $x \in \R$, let $[x]_+$ denote the positive part of $x$: $[x]_+ = x$ if $x > 0$,
and $[x]_+ = 0$ otherwise.
Then the expression 
\[
\left[\sum_{i,j} a_{i'j} p_{ij}(t) - \sum_{i,j} a_{ij} p_{ij}(t)\right]_+
\]
can be interpreted as the regret of the first player from not 
having played action $i' \in \Pure{A}$ throughout the entire past history of play. 
It is (the positive part of) the difference between player A's actual past 
payoff\footnote{Note that $\sum_{i,j} a_{ij} p_{ij}(t)$ and $\sum_{i,j} b_{ij} p_{ij}(t)$
are the players' average payoffs in their play through
time $t$.} and the payoff she would have received if she always played $i'$, given 
that player B would have played the same way as she did.
Similarly, $[\sum_{i,j} b_{ij'} p_{ij}(t) - \sum_{i,j} b_{ij} p_{ij}(t)]_+$ 
is the regret of the second player from not having played $j' \in \Pure{B}$. 
This regret notion is sometimes called \emph{unconditional} or \emph{external regret} 
to distinguish it from the \emph{internal} or 
\emph{conditional regret}\footnote{Conditional regret is the regret from not having 
played an action $i'$ whenever a certain action $i$ has been played, that is, 
$[\sum_j a_{i'j} p_{ij} - \sum_j a_{ij} p_{ij}]_+$ for some fixed $i \in \Pure{A}$.}.
In this context the set of CCE can be interpreted as 
the set of joint probability distributions with non-positive regret. 

It has been shown that there are learning algorithms with no regret, 
that is, such that asymptotically the regret of players playing according 
to such algorithm is non-positive for all their actions. 
Dynamically this means that if both players in a two-player game use a no-regret learning algorithm, 
the empirical joint probability distribution
of actions taken by the players converges to (a subset of) the set of CCE 
(\textit{not} necessarily to a certain point in this set). 

The concept of no-regret learning (also known as \emph{universal consistency}, 
see~\cite{Fudenberg1995}) and the first such learning algorithms 
have been introduced in~\cite{Blackwell1954,Hannan1957}. 
More such algorithms have been found later on and moreover algorithms with 
asymptotically non-positive \emph{conditional} regrets 
have been found (see, for example,~\cite{Foster1997,Hart2000,Hart2001a};
for good surveys see~\cite{Young2005b,Hart2005}). 

We now show that continuous-time FP converges to a subset of CCE, namely the subset 
for which equality holds for at least one $(i',j') \in \Pure{A} \times \Pure{B}$ 
in  (\ref{eq:hanan}). 

\begin{theorem}\label{thm:lim_set}
Every trajectory of FP dynamics~\eqref{eq:fp}
in a bimatrix game $(A,B)$ converges to a subset of the set of CCE, 
the set of joint probability distributions 
$P = (p_{ij})$ over $\Pure{A} \times \Pure{B}$ such that 
for all $(i',j') \in \Pure{A} \times \Pure{B}$
\begin{equation}
\sum_{i,j} a_{i'j} p_{ij} \leq \sum_{i,j} a_{ij} p_{ij} \quad \text{ and }\quad
\sum_{i,j} b_{ij'} p_{ij} \leq \sum_{i,j} b_{ij} p_{ij}, 
\label{eq:hanan}
\end{equation}
where {\em equality} holds {\em for at least one} $(i',j') \in \Pure{A} \times \Pure{B}$.
In other words, FP dynamics asymptotically 
leads to non-positive (unconditional) regret for both players. 
\end{theorem}

\begin{remark}
\begin{enumerate}
\item Note that an FP orbit $(p(t),q(t))$, $t \geq 1$, gives rise to a joint 
probability distribution $P(t) = (p_{ij}(t))$ via $p_{ij}(t) = \frac{1}{t} \int_0^t x_i(s)y_j(s)ds$.
When we say that FP converges to a certain set
of joint probability distributions, we mean that $P(t)$ obtained this way converges to this set.

\item In~\cite{Monderer1997} a stronger result is proved:
continuous-time FP is \q{belief affirming} or \q{Hannan-consistent}.
This means that it leads to asymptotically non-positive
unconditional regret for the player following it, \emph{irrespective of her opponent's play} 
(even if the opponent is playing according to a different algorithm). 
We will only need the weaker statement and provide our own proof for the reader's 
convenience.
\end{enumerate} 
\end{remark}

\begin{proof}[Proof of Theorem~\ref{thm:lim_set}]
We assume that we have an orbit of~\eqref{eq:fp}, $(p(t),q(t))$, $t \geq 0$.
Recall that $p(t) = \frac{1}{t}\int_0^t x(s)\, ds$
and $q(t) = \frac{1}{t}\int_0^t y(s)\, ds$, where 
$x \colon [0,\infty) \rightarrow \Siga$ and  $y \colon [0,\infty) \rightarrow \Sigb$
are measurable functions representing the players' strategies at any time $t \geq 0$, 
so that $x(t) \in \BRa( q(t) )$ and $y(t) \in \BRb( p(t) )$ for $t \geq 1$.

By the envelope theorem (see, for example,~\cite{Sydsaeter2008}), 
for $\bar p \in \BRa(q)$ we have that
\[
\frac{d\bar A(q)}{dq} = \frac{\partial \poa(p,q)}{\partial q}\bigg|_{p = \bar p} = \bar p A.
\]
Therefore, since $x(t) \in \BRa(q(t))$ for $t \geq 1$,
\[
\frac{d}{dt} \left(t \bar A( q(t) )\right) 
= \bar A(q(t)) + t \frac{d}{dt}\left(\bar A(q(t))\right)
= \bar A(q(t)) + t \cdot x(t)\cdot A\cdot \frac{dq(t)}{dt}.
\]
Using~\eqref{eq:fp}
and $\bar A(q(t)) = x(t) \cdot A \cdot q(t)$, it follows that
\[
\frac{d}{dt} \left(t \bar A( q(t) )\right)  
= \bar A(q(t)) +  x(t) \cdot A \cdot \left(y(t)-q(t) \right)
= x(t) \cdot A \cdot y(t)
\]
for $t \geq 1$. We conclude that for $T > 1$,
\[
\int_1^T x(t) \cdot A \cdot y(t) \, dt = T \bar A(q(T)) - \bar A (q(1)),
\]
and therefore
\[
\lim_{T\rightarrow\infty} \left( \frac{1}{T}\left( \int_0^T x(t) \cdot A \cdot y(t) \, dt\right) - \bar A(q(T))\right) = 0. 
\]
Note that
\[
\frac{1}{T} \int_0^T x(t) \cdot A \cdot y(t) \, dt = \sum_{i,j} a_{ij}p_{ij}(T),
\]
where $P(T) = (p_{ij}(T))$ is the empirical joint distribution 
of the two players' play through time $T$. 
On the other hand, 
\[\bar A(q(T)) = \max_{i'} \sum_j a_{i'j} q_j(T) 
= \max_{i'} \sum_{i,j} a_{i'j} p_{ij}(T).
\]
Hence, 
\[
\lim_{T\rightarrow\infty} \left( \sum_{i,j} a_{ij}  p_{ij}(T) - \max_{i'} \sum_{i,j} a_{i'j} p_{ij}(T) \right) = 0.
\]
By a similar calculation for $B$, we obtain
\[
\lim_{T\rightarrow\infty} \left( \sum_{i,j} b_{ij}  p_{ij}(T) - \max_{j'} \sum_{i,j} b_{ij'} p_{ij}(T) \right) = 0.
\]
It follows that any FP orbit converges to the set of CCE.
Moreover, these equalities imply that for a sequence
$t_k \to \infty$  so that $p_{ij}(t_k)$ converges, 
there exist $i',j'$ so that $\sum_{i,j} (a_{ij} - a_{i'j})  p_{ij}(t_k) \to 0$ 
and $\sum_{i,j} (b_{ij} -b_{ij'}) p_{ij}(t_k) \to 0$ as $k \to \infty$, proving convergence
to the claimed subset.
\end{proof}

Let us denote the average payoffs through time $T$ along an FP orbit as
\[
\hpoa(T) = \frac{1}{T} \int_0^T x(t) \cdot A \cdot y(t) \, dt  \quad \text{and} \quad
\hpob(T) = \frac{1}{T} \int_0^T x(t) \cdot B \cdot y(t) \, dt.
\]
As a corollary to the proof of the previous theorem we get the following useful result.

\begin{theorem}\label{thm:lim_av_payoff}
In any bimatrix game, along every orbit of FP dynamics we have
\begin{align*}
\lim_{T\rightarrow\infty} \left( \hpoa(T) - \bar A(q(T))\right) = 
\lim_{T\rightarrow\infty} \left( \hpob(T) - \bar B(p(T))\right) = 0.
\end{align*}
\end{theorem}

\begin{remark}
This formulation of the result shows why in~\cite{Monderer1997}
this property is called \q{belief affirming}. Since $\bar A(q(T))$ and $\bar B(p(T))$ 
can be interpreted as the players' expected payoffs given their respective
opponent's play $q(T)$ and $p(T)$, the above theorem says
that the difference between expected and actual average payoff
of each player vanish, so that asymptotically their \q{beliefs} are \q{confirmed}
when playing according to FP.
\end{remark}

\section{FP vs.~Nash equilibrium payoff} 
\label{sec:fp_vs_ne}

In this section we investigate the average payoff to players in a two-player 
game along the orbits of FP dynamics and compare it to the Nash equilibrium payoff 
(in particular, in games with a unique, completely mixed Nash equilibrium). 
We show that in contrast to the usual assumption that players should primarily attempt 
to play Nash equilibrium and that learning algorithms converging to Nash equilibrium are desirable, 
the payoff along FP orbits can in some games be better on average,
or even at all times Pareto dominate the Nash equilibrium payoff. 

Moreover, we demonstrate that to every bimatrix game with unique, 
completely mixed Nash equilibrium, there is a dynamically equivalent game for which 
this superiority of FP over Nash equilibrium holds. 

Throughout the rest of this section we will assume that all the games under consideration
have a unique, completely mixed Nash equilibrium point $\nash$ 
(it is a well-known fact that in such a game, both players 
necessarily have the same number of strategies). 
A first simple situation in which FP can improve upon such a Nash equilibrium is 
given by the following direct consequence of Theorem~\ref{thm:lim_av_payoff}. 

\begin{proposition} \label{prop:fp_opt}
Let $(A,B)$ be a bimatrix game with unique, completely mixed Nash equilibrium $\nash$. 
If $\bar A(q) \geq \bar A(\nashb)$ and $\bar B(p) \geq \bar B(\nasha)$ for all $(p,q) \in \Sigma$, 
then asymptotically the average payoff along FP orbits is 
greater than or equal to the Nash equilibrium payoff (for both players). 
\end{proposition}

\begin{remark}
The hypothesis of this proposition, $\bar A(q) \geq \bar A(\nashb)$ and 
$\bar B(p) \geq \bar B(\nasha)$ for all $(p,q) \in \Sigma$, means that 
\[
\poa\nash = \min_{q \in \Sigb} \max_{p \in \Siga} p A q 
\quad \text{and} \quad 
\pob\nash = \min_{p \in \Siga} \max_{q \in \Sigb} p B q,
\]
that is, the Nash equilibrium payoff equals the minmax payoff of the players. 
For a non-zero-sum game this is a rather strong assumption, suggesting an unusually
bad Nash equilibrium in terms of payoff. However, as we will show in the next result, at 
least from a dynamical point of view, the situation is not at all exceptional.
\end{remark}

\begin{theorem}\label{thm:exist_le}
Let $(A,B)$ be an $n \times n$ bimatrix game with unique, completely mixed Nash equilibrium $\nash$.
Then there exists a linearly equivalent game $(A',B')$, for which
$\bar {A'}(q) > \bar {A'}(\nashb)$ and $\bar B'(p) > \bar B'(\nasha)$ 
for all $p \in \Siga \setminus \{ \nasha\}$ and $q \in \Sigb \setminus \{\nashb\}$.
\end{theorem}

This result states that every bimatrix game with unique, 
completely mixed Nash equilibrium is linearly 
equivalent to one in which 
players are better off playing FP than playing 
the (unique) Nash equilibrium strategy. 
In the proof we will need the following lemma.

\begin{lemma}\label{lem:inteq}
Let $(A,B)$ be an $n \times n$ bimatrix game with unique, completely mixed Nash equilibrium $\nash$.
Then for each $k \in \Pure{A}$, 
$L^A_k \coloneqq \left(\bigcap_{i \neq k} \Rega{i}\right) \setminus \Rega{k}$ is non-empty. 
More precisely, $L^A_k$ is a ray from $\nashb$ in the direction $v^k$, 
such that any $(n-1)$ of the $n$ vectors $v^1,\ldots,v^n$ form a basis for the space 
$\{v \in \R^n \colon \sum_i v_i = 0 \}$. The analogous statement applies to
$L^B_l \coloneqq \left(\bigcap_{j \neq l} \Regb{j} \right) \setminus \Regb{l}$, $l \in \Pure{B}$.
\end{lemma}

\begin{proof}
Define the projection 
\[
\pi \colon \left\{x \in \R^n : \sum_i x_i = 1\right\} \to \R^{n-1},\quad \pi(x_1,\ldots,x_n) = (x_1, \ldots, x_{n-1}),
\]
and note that $\pi$ is invertible with inverse 
\[
\pi^{-1}(y) = (y_1, \ldots, y_{n-1}, 1 - \sum_{k=1}^{n-1}y_k).
\]
For $q \in \Sigb$ we have that $\sum_{k=1}^n q_k = 1$ and therefore
\[
(A q)_i - (A q)_j 
= \sum_{k=1}^n (a_{ik} - a_{jk}) q_k
= \sum_{k=1}^{n-1} (a_{ik}-a_{jk}-a_{in}+a_{jn} ) q_k + (a_{in}-a_{jn}), 
\]
and we define the affine map $P\colon \R^{n-1} \to \R^{n-1}$ by
\[
P_l (x) = \sum_{k=1}^{n-1} (a_{l,k}-a_{l+1,k}-a_{l,n}+a_{l+1,n} ) x_k + (a_{l,n}-a_{l+1,n}),
\]
for $l = 1, \ldots, n-1$ (that is, $P_l(x) = A(\pi^{-1}(x))_l - A(\pi^{-1}(x))_{l+1}$).

Recall from Lemma~\ref{lem:ne_int} that for $(p,q) \in \Sigma$, 
$q = \nashb$ if and only if $(A q)_i = (A q)_j$ for all $i,j$, 
and $p = \nasha$ if and only if $(p B)_i = (p B)_j$ for all $i,j$. 
It follows that 
\[
P(x) = 0 \quad \text{ if and only if } \quad x = \pi(\nashb).
\]
In particular, the affine map $P$ is invertible and there is a unique vector 
$v^1 \in \{v \in \R^n : \sum_i v_i = 0\}$, 
such that $P(\pi(\nashb +v^1)) = w^1 \coloneqq (-1, 0, \ldots, 0)^\top$. 
Since $\nashb$ is in the interior of $\Sigb$, $x^1 = \nashb+s \cdot v^1 \in \Sigb$ 
for sufficiently small $s > 0$, and
we have that $P(\pi(x^1)) = (-s, 0, \ldots, 0)^\top$. 
By the definition of $P$, this means that 
\[
(A x^1)_1 < (A x^1)_2 = (A x^1)_3 = \cdots = (A x^1)_n.
\] 
Hence $x^1 \in L^A_1 = \left(\bigcap_{k \neq 1} \Rega{k}\right) \setminus \Rega{1}$. 
Note also that every $x \in L^A_1$ is of the form $\nashb + s \cdot v^1$ 
for some $s > 0$, that is, $L^A_1$ is a ray from the point $\nashb$. 

For $1 < k < n$, let $w^k$ be the vector in $\R^n$ with $(k-1)$th 
and $k$th entries equal to $1$ and $-1$ respectively, 
and all other entries equal to $0$. Then choose $v^k$ such that $P(\pi(\nashb+v^k)) = w^k$. 
Again for sufficiently small $s > 0$, we get $x^k = \nashb + s \cdot v^k \in L^A_k$.
Finally, for $k = n$, let $w^k = (0, \ldots, 0, 1)$ and proceed 
as above to get $v^n$ and $x^n = \nashb + v^n \in L^A_n$. 

Writing the affine map $P$ as $P(x) = M x + b$ for some invertible 
matrix $M \in \R^{(n-1)\times(n-1)}$ and $b \in \R^{n-1}$, we get  
\[
w^k = P(\pi(\nashb + v^k)) =  P(\pi(\nashb)) + M (v^k_1, 
\ldots, v^k_{n-1})^\top = M (v^k_1, \ldots, v^k_{n-1})^\top, \quad k = 1,\ldots,n.
\]
Since any $n-1$ of the vectors 
\[
w^1 = \begin{pmatrix} -1 \\ 0 \\ 0 \\ \vdots \\ 0   \end{pmatrix},\quad 
w^2 = \begin{pmatrix} 1 \\ -1 \\ 0 \\ \vdots \\ 0   \end{pmatrix},\quad
\ldots,\quad
w^{n-1} = \begin{pmatrix} 0  \\ \vdots \\ 0 \\ 1 \\ -1   \end{pmatrix},\quad
w^n = \begin{pmatrix} 0 \\ \vdots \\ 0 \\ 0 \\ 1   \end{pmatrix}
\]
are linearly independent and $M$ is invertible, 
it follows that any $n-1$ of the vectors $v^1,\ldots,v^n$ are linearly independent, as claimed. 

The same argument applied to the matrix $B^\top$ 
shows the analogous result for $L^B_l$, $l = 1,\ldots,n$, which finishes the proof. 
\end{proof}

\begin{proof}[Proof of Theorem~\ref{thm:exist_le}]
Let $A' \in \R^{n\times n}$, such that $a'_{ij} =a_{ij} + c_j$ for some $c = (c_1, \ldots, c_n) \in \R^n$. 
Then for any $q \in \Sigb$, 
\begin{equation}\label{eq:adash_a}
\bar {A'}(q) = \max_i\left(A' q\right)_i 
= \max_i\left(A q + \sum_{j=1}^n c_j q_j \cdot
\begin{pmatrix} 1 \\ \vdots \\ 1 \end{pmatrix}\right) _i = \bar A(q) + c \cdot q. 
\end{equation}

Observe that, restricted to $\Rega{k}$, level sets of $\bar A$ are precisely the 
$(n-2)$-dimensional hyperplane pieces in $\Sigb$ orthogonal to $\underline {a}_k$, the 
$k$th row vector of $A$:
\[
q - \tilde q \perp \underline {a}_k \Leftrightarrow
q \cdot \underline {a}_k = \tilde q \cdot \underline {a}_k \Leftrightarrow
\max_j(A q)_j = \max_j (A \tilde q)_j \qquad \text{for } q,\tilde q\in \Rega{k}.
\]
So all level sets of $\bar A$ restricted to $\Rega{k}$ are parallel hyperplane pieces. 
Figure~\ref{fig:H} illustrates this situation for the case $n = 3$.

\begin{figure}
\centering
\includegraphics[width = 0.5\textwidth]{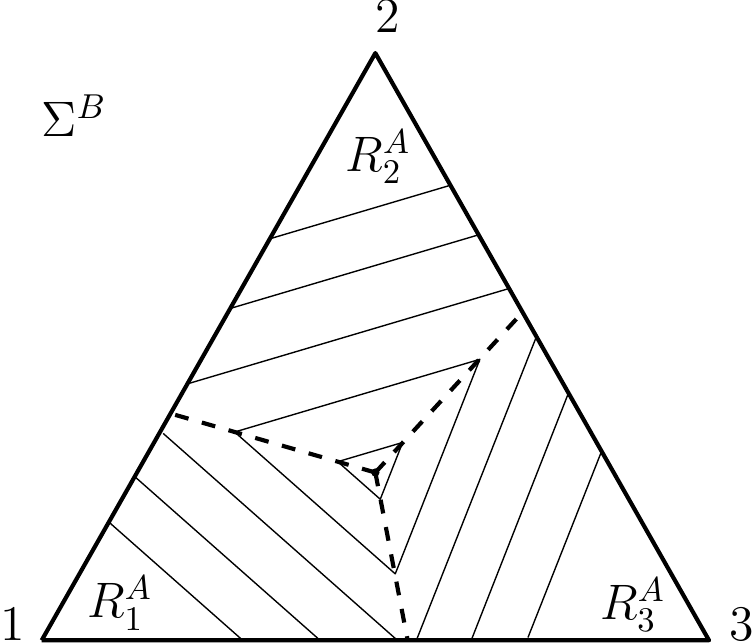}
\caption{(Proof of Theorem~\ref{thm:exist_le}) 
Level sets for $\bar A$ restricted to each region $\Rega{i}$ 
are parallel line segments in $\Sigb$ (in a $3 \times 3$ game).}
\label{fig:H}
\end{figure}

By Lemma~\ref{lem:inteq} we can choose $n$ points $Q_1, \ldots, Q_n \in \Sigb$ such that
\[
Q_k \in L^A_k = \left(\bigcap_{i \neq k} \Rega{i} \right) \setminus \Rega{k}. 
\]
Each point $Q_k$ is in the relative interior of the line segment $L^A_k \subset \Sigb$.
This line segment has endpoint $\nashb$ and is adjacent to all 
of the regions $\Rega{i}$, $i \neq k$. 
By the same lemma, $Q_1 - \nashb, \ldots, Q_{n-1}-\nashb$ form a basis for $\{v \in \R^n \colon \sum_k v_k = 0 \}$. 
Therefore, the vectors $Q_1, \ldots, Q_n$ form a basis for $\R^n$. 

It follows that one can choose $c = (c_1,\ldots, c_n) \in \R^n$, such that 
\[
c \cdot Q_1 + \bar A(Q_1) = \cdots = c \cdot Q_n + \bar A(Q_n),
\]
and hence by~\eqref{eq:adash_a},
\[
\bar {A'}(Q_1) = \cdots = \bar {A'}(Q_n).
\]
Then level sets of $\bar {A'}$ are boundaries of $(n-1)$-dimensional simplices centred at $\nashb$ 
(each similar to the simplex with vertices $Q_1, \ldots, Q_n$). 

\medskip

Now we show that $\nashb$ is a minimum for $\bar {A'}$. 
By uniqueness of the completely mixed Nash equilibrium and Lemma~\ref{lem:ne_int},
$A$ has a row vector which is not a multiple of $(1, \ldots, 1)$. 
Therefore, there exists a vector $v = (v_1, \ldots, v_n) \in \R^n$  with $\sum_k v_k = 0$,
such that at least one of the entries of $A v$ is positive. 
Let $r(t) = \nashb + t\cdot v$, $t \geq 0$, be a ray from $\nashb$ in $\Sigb$. Then for $t_2 > t_1$ we get
\[
\bar {A'}(r(t_2)) - \bar {A'}(r(t_1)) 
= \max_j (A \nashb + t_2 A v)_j - \max_j(A \nashb + t_1 A v)_j 
= (t_2 - t_1) \max_j (A v)_j > 0. 
\]
So, along some ray from $\nashb$, $\bar {A'}$ is increasing. By the spherical structure of the level sets, 
this implies that $\bar {A'}$ is increasing along every ray from $\nashb$. 
Hence $\bar {A'}(\nashb) \leq \bar {A'}(q)$ for every $q \in \Sigb$ with equality only for $q = \nashb$. 

The same reasoning shows that one can choose $d_1, \ldots, d_n \in \R$ and 
$B' \in \R^{n \times n}$, $b'_{ij} = b_{ij} + d_i$, 
such that $\bar {B'}(\nasha) \leq \bar {B'}(p)$ for every $p \in \Siga$ with equality only for $p = \nasha$.
\end{proof}

The previous results, Theorem~\ref{thm:exist_le} and Proposition~\ref{prop:fp_opt}, 
assert that every game possesses a dynamically equivalent version, 
in which FP Pareto dominates 
Nash equilibrium play. This shows that dynamical equivalence
does not in general preserve the global payoff structure 
of a game, since there are clearly games
for which Pareto dominance of FP over 
Nash equilibrium does not hold a priori.

In the famous Shapley game or variants of it~\cite{Shapley1964,Sparrow2007,VanStrien2010},
FP typically converges to a limit cycle, known as a 
Shapley polygon~\cite{Gaunersdorfer1995},
and usually the payoff along this polygon is greater than the Nash 
equilibrium payoff in some parts of the cycle, 
and less in others. On average, this can be still preferable for both 
players compared to playing Nash equilibrium, 
if they aim to maximise their time-average payoffs. In a similar setting, this 
has been previously observed in~\cite{Gaunersdorfer1995}. We will show 
an example of this situation in the next section. 

In fact, the proof of Theorem~\ref{thm:exist_le} shows that
the unique, completely mixed Nash equilibrium $\nash$ can never 
be an isolated payoff-maximum, since there are always 
directions from $\nashb$ in $\Sigb$ and from $\nasha$ in $\Siga$
along which $\bar A$ and $\bar B$ are non-decreasing. Heuristically 
one would therefore expect
that FP typically improves upon Nash equilibrium 
in at least parts of any limit cycle. 
In Section~\ref{sec:fp_worse} we will demonstrate that this is not always the case:
there are games in which FP 
typically produces a lower average payoff than Nash equilibrium.

\section{FP better than Nash equilibrium: an example} 
\label{sec:ex_family}

Consider the one-parameter family of 
$3 \times 3$ bimatrix games $(A_\beta, B_\beta)$, $\beta \in (0,1)$, given by 
\begin{equation}\label{eq:shapley_family}
A_\beta = \begin{pmatrix} 1 & 0 & \beta \\ \beta & 1 & 0 \\ 0 & \beta & 1\end{pmatrix}, \qquad
B_\beta = \begin{pmatrix} -\beta & 1 & 0 \\ 0 & -\beta & 1 \\1 & 0 & -\beta\end{pmatrix}.
\end{equation}
This family can be viewed as a generalisation 
of Shapley's game~\cite{Shapley1964}. In~\cite{Sparrow2007,VanStrien2010}, 
FP dynamics of this family of games has been studied extensively, 
and the system has been shown to give rise to a very rich chaotic dynamics
with many unusual and remarkable dynamical features. 
The game has a unique, completely mixed Nash equilibrium $\nash$, where 
$\nasha = (\nashb)^\top = (\frac{1}{3},\frac{1}{3},\frac{1}{3})$,
which yields the respective payoffs 
\[
\poa\nash = \frac{1 + \beta}{3}  \quad\text{and}\quad \pob\nash = \frac{1-\beta}{3}.
\]
To check the hypothesis of Proposition~\ref{prop:fp_opt}, 
let $q = (q_1, q_2, q_3)^\top \in \Sigb$, then
\begin{align*}
\bar A(q) 
&= \max \left\{ q_1+ \beta q_3, q_2 + \beta q_1, q_3 + \beta q_2 \right\}\\
&\geq \frac{1}{3} ( (q_1+ \beta q_3)+( q_2 + \beta q_1)+( q_3 + \beta q_2) )\\
& = \frac{1}{3} (q_1 + q_2 + q_3) (1 + \beta) \\
& = \frac{1+\beta}{3} \\
& = \poa\nash = \bar A(\nashb).
\end{align*}
Moreover, equality holds if and only if 
\[
  q_1+ \beta q_3 =  q_2 + \beta q_1 = q_3 + \beta q_2,
\]
which is equivalent to $q_1 = q_2 = q_3$, that is, $q = \nashb$. We conclude that 
$\bar A(q) > \bar A(\nashb)$ for all $q \in \Sigb\setminus\{\nashb\}$, and 
by a similar calculation, $\bar B(p) > \bar B (\nasha)$ for all $p \in \Siga\setminus\{\nasha\}$. 
As a corollary to Proposition~\ref{prop:fp_opt} we get the following result.

\begin{theorem}\label{thm:beta_payoff}
Consider the one-parameter family of bimatrix games $(A_\beta, B_\beta)$ 
in~\eqref{eq:shapley_family} for $\beta \in (0,1)$. 
Then any (non-stationary) FP orbit Pareto dominates constant Nash equilibrium play in the long run, 
that is, for large times $t$ we have
\[
\hpoa(t) > \poa\nash \quad \text{and} \quad \hpob(t) > \pob\nash.
\]
\end{theorem}

In fact, one can say more: There is a $\beta \in (0,1)$
such that FP has an attracting closed orbit 
(the so-called \q{anti-Shapley orbit}~\cite{Sparrow2007,VanStrien2010})
along which FP Pareto dominates Nash equilibrium \emph{at all times}. 
In other words, both players are receiving a higher payoff than at Nash equilibrium at 
any time along this orbit. We omit the details of the proof: 
techniques developed in~\cite{Krishna1998,Rosenmuller1971} 
can be used to analyse FP along this orbit, whose existence
was shown in~\cite{Sparrow2007}.
In particular, the times spent in each region $\Regb{j} \times \Rega{i}$ along the orbit can be 
worked out explicitly, which can be directly applied to obtain average payoffs.

\begin{remark}
In fact, FP also improves upon the set of \q{correlated equilibria} 
in this family of games. 
The famous notion of correlated equilibrium, introduced in~\cite{Aumann1974, Aumann1987}, 
is defined as follows. 
A joint probability distribution $P = (p_{ij})$ over 
$S = \Pure{A} \times \Pure{B}$ is a 
\emph{correlated equilibrium (CE)} for the bimatrix game $(A,B)$ if 
\[
\sum_k a_{i'k} p_{ik} \leq \sum_k a_{ik} p_{ik} \quad \text{and} 
\quad \sum_l b_{lj'} p_{lj} \leq \sum_k b_{lj} p_{lj}
\]
for all $i,i' \in \Pure{A}$ and $j,j' \in \Pure{B}$. 
One interpretation of this notion is similar to that of the 
CCE (see paragraph after Definition~\ref{def:cce}), with 
the notion of \q{(unconditional) regret} replaced by the finer notion of \q{conditional regret}.
If we think of $P$ as the empirical distribution of play up to a certain time 
for two players involved in repeatedly or continuously playing a given game, 
then $P$ is a CE if neither player regrets
not having played a strategy $i'$ (or $j'$) whenever she actually played $i$ (or $j$). 
In other words, the average payoff to player A would not be higher, if she would
have played $i'$ at all times when she actually played $i$ throughout the history of play 
(assuming her opponent's behaviour unchanged), and the same for player B. 

One can check that the set of Nash equilibria is always contained
in the set of CE, which in turn is always contained in the set 
of CCE. In the game $(A_\beta,B_\beta)$ 
in Theorem~\ref{thm:beta_payoff}, the Nash equilibrium $\nash$ 
is also the unique CE, which can be checked by direct computation.
Hence our result shows that in this case, FP also improves upon 
CE in the long run.
\end{remark}

\section{FP can be worse than Nash equilibrium} 
\label{sec:fp_worse}

We have seen that in many games FP
improves upon Nash equilibrium in terms of payoff. 
Moreover, we have shown that for any bimatrix game 
with unique, completely mixed Nash equilibrium, 
linear equivalence can be used to obtain dynamically equivalent examples in which
FP Pareto dominates Nash equilibrium. 
In this section we investigate the converse possibility of 
FP having lower payoff than Nash equilibrium. Again 
we restrict our attention to $n\times n$ games with 
unique, completely mixed Nash equilibrium. 

Let us define the \emph{sub-Nash payoff cones}, 
the set of those mixed strategies of player A, 
for which the best possible payoff to player B is not greater than Nash equilibrium payoff,
\[
P_B^- = \{ p \in \Siga : \max_i (p B)_i \leq \max_i (\nasha B)_i  \}, 
\]
and similarly
\[
P_A^- = \{ q \in \Sigb : \max_j (A q)_j \leq \max_j (A \nashb)_j  \}. 
\]
By adding suitable constants to the player's payoff matrices we can assume 
without loss of generality that $\poa\nash = \pob\nash = 0$. 
Then one can see that 
\[
P_B^- = (B^\top)^{-1}(\R^n_-) \cap \Siga \quad \text{and} \quad P_A^- = A^{-1}(\R^n_-) \cap \Sigb,
\]
where $\R^n_-$ denotes the quadrant of $\R^n$ with all coordinates non-positive, 
and by $(B^\top)^{-1}$ and $A^{-1}$ we mean the pre-images 
under the linear maps $B^\top, A \colon \R^n \to \R^n$. 
Therefore, $P_B^-$ and $P_A^-$ are (closed) convex cones in $\Siga$ 
and $\Sigb$ with apexes $\nasha$ and $\nashb$ respectively. 

Now an orbit of FP is Pareto dominated by 
Nash equilibrium if and only if it (or its part for $t \geq t_0$ for some $t_0$) 
is contained in the interior of
$P_B^- \times P_A^-$. This shows that a
result like Theorem~\ref{thm:exist_le} with the roles of Nash equilibrium 
and FP reversed cannot hold: 
if a game has an FP orbit whose projections 
to $\Siga$ and $\Sigb$ are not both contained in some 
convex cones with apexes $\nasha$ and $\nashb$, then for any 
linearly equivalent game, along this orbit there are times at which one of the players 
enjoys higher payoff than Nash equilibrium payoff. 
In order to find FP orbits along which 
payoffs are permanently lower than Nash equilibrium payoff, 
one therefore needs to find orbits contained in a 
halfspace (whose boundary plane contains the Nash equilibrium).
The following lemma ensures that one can then obtain 
a linearly equivalent game with $P_B^- \times P_A^-$ containing this orbit. 

\begin{lemma} \label{lem:sub_nash}
Let $(A,B)$ be any $n \times n$ bimatrix game with 
unique, completely mixed Nash equilibrium $\nash$. 
Let $H_A$ and $H_B$ be open halfspaces such that $\nasha \in \partial H_A$ and 
$\nashb \in \partial H_B$. Let further $C_A$ and $C_B$ be closed convex polyhedral cones with
non-empty interior and apexes $\nasha$ and $\nashb$ respectively, such that
\begin{itemize}
\item $C_A \setminus \{\nasha\} \subset \Siga \cap H_A$ and 
$C_B \setminus \{\nashb\} \subset \Sigb \cap H_B$,
\item $C_A$ ($C_B$) contains exactly one of the line segments $L^B_i\setminus \{\nasha\}$  ($L^A_j \setminus \{\nashb\}$) in its interior,
\item $C_A$ ($C_B$) has exactly $n-1$ extreme rays, each lying in
the interior of one of $\Regb{j}$ ($\Rega{i}$), such that each $\Regb{j}$ ($\Rega{i}$) contains at most one such ray.
\end{itemize}
Then there exists a linearly equivalent game 
$(A', B')$, such that $P_{B'}^- = C_A$ and $P_{A'}^- = C_B$. 
\end{lemma}

\begin{proof}[Proof of Lemma~\ref{lem:sub_nash}]
The proof follows the same line of argument as the proof of 
Theorem~\ref{thm:exist_le} and therefore we will refer to that proof.
Note that in the proof of Theorem~\ref{thm:exist_le}, to any given game
we constructed a linearly equivalent game with $P_A^- = \{E_B\}$ and $P_B^- = \{E_A\}$.
To prove the lemma, without loss of generality assume that $L^A_n \subset C_B$, 
which implies that 
$\partial C_B$ has non-empty intersection with the interior of each of $\Rega{i}$ for $i \neq n$.
We can then pick $n-1$ points $Q_i \in \partial C_B \cap \inter (\Rega{i})$ on the $n-1$ extreme rays 
and similarly
to the proof of Theorem~\ref{thm:exist_le} prescribe the $n$ linear equations
$\bar A' (\nashb) = \bar A' (Q_1) = \ldots = \bar A' (Q_{n-1})  = 0$, where $A'$ is again the
matrix obtained from $A$ by adding constants $c_1, \ldots, c_n$ to its columns.
This has a unique solution for $c = (c_1, \ldots, c_n)$, since $\nashb, Q_1, \ldots , Q_{n-1}$
form a basis for $\R^n$. Then, by construction, $\bar A'$ is $0$ on $\partial C_B$. 
Because of the structure of the level sets of $\bar A'$ worked out in the proof of 
Theorem~\ref{thm:exist_le}, this implies that either $C_B$ or the closure of its complement in $\Sigb$ is the set on which $\bar A' \leq 0 = \bar A'(\nashb)$, that is, $P_{A'}^-$. But since both $C_B$ and $P_{A'}^-$ are convex, it follows that $P_{A'}^- = C_B$, and analogously one can find a 
linearly equivalent matrix $B'$ so that $P_{B'}^- = C_A$.
\end{proof}
 
 \medskip

By Lemma~\ref{lem:sub_nash}, to find an example of a game with an 
orbit which is Pareto worse than Nash equilibrium, 
it suffices to find a game with an orbit whose projections to $\Siga$ and $\Sigb$ 
are completely contained in suitable convex cones with apexes $\nasha$ and $\nashb$ respectively.
One can then construct a linearly equivalent game, for which this 
orbit is actually contained in the sub-Nash payoff cones. 
We will demonstrate one such example in the $3 \times 3$ case, which we 
obtained by numerically randomly generating $3 \times 3$ games
and testing large numbers of initial conditions to detect orbits of the desired type. 

Observe that by convexity of the preference regions $\Rega{i}$, 
a halfspace in $\Sigb$ whose boundary
contains the (unique, completely mixed) Nash equilibrium 
contains at most two of the three rays $L^A_i$, $i=1,2,3$. 
The same holds for a halfspace in $\Siga$ and the rays $L^B_j$, $j=1,2,3$. 
Hence an orbit entirely contained in such halfspace never crosses at 
least one of these lines for each player. 

\begin{example}
Let the bimatrix game $(A,B)$ be given by
\begin{align*}
A = \begin{pmatrix}-1.353259 & -1.268538 &  2.572738 \\ 0.162237 & -1.800824 &  1.584291 \\ -0.499026 & -1.544578 & 1.992332 
\end{pmatrix}, ~
B = \begin{pmatrix}-1.839111 & -2.876997 & -3.366031 \\ -4.801713 & -3.854987 & -3.758662 \\ 6.740060 & 6.590451 & 6.898102
\end{pmatrix}.
\end{align*}
This bimatrix game has a unique Nash equilibrium $\nash$ with
\[
\nasha \approx (0.288 , 0.370 , 0.342 ),\quad \nashb \approx (0.335 , 0.327 , 0.338)^\top . 
\]
The matrices $A$ and $B$ are chosen in such a way that the Nash 
equilibrium payoffs are both normalised to zero: $\poa\nash = \pob\nash = 0$. 
Numerical simulations suggest that FP 
has a periodic orbit as a stable limit cycle, 
which attracts almost all initial conditions. 
This trajectory forms an octagon in the four-dimensional space $\Sigma = \Siga \times \Sigb$, 
it is depicted in Figure~\ref{fig:per_orbit}. 
The orbit follows an 8-periodic itinerary of the form
\[
(2,1) \rightarrow  (2,2) \rightarrow  (3,2) \rightarrow  (3,3) \rightarrow  
(1,3) \rightarrow  (1,2) \rightarrow  (1,1) \rightarrow (3,1) \rightarrow (2,1) .
\]
(That is, there is a strictly increasing sequence of times $(t_i)_{i \geq 1}$
such that $(p(t),q(t)) \in \Regb{2} \times \Rega{1}$ for $t \in (t_1,t_2)$,
$(p(t),q(t)) \in \Regb{2} \times \Rega{2}$ for $t \in (t_2,t_3)$, 
$(p(t),q(t)) \in \Regb{3} \times \Rega{2}$ for $t \in (t_3,t_4)$, etc.)
Note that the second player's best response never changes from 1 to 2, 
nor vice versa. Similarly, for player A the best response
never directly changes between 1 and 3 without an intermediate step through 2. 
Moreover, it can be seen from Figure~\ref{fig:per_orbit} that 
the projections of the periodic orbit to $\Siga$ and $\Sigb$ lie in halfplanes 
whose boundaries contain the points $\nasha$ and $\nashb$ respectively. 
Hence Lemma~\ref{lem:sub_nash} allows us
to choose the matrices $A$ and $B$ such that this orbit
lies completely in $P_B^- \times P_A^-$, so that the payoffs 
to both players are permanently worse than Nash equilibrium payoff. 
Figure~\ref{fig:per_orbit_po} shows the (negative) payoffs to both 
players along several periods of the orbit and the higher (zero) Nash equilibrium payoff. 

\begin{figure}[t]
\centering
\includegraphics[width = 0.7\textwidth, trim = 0 200 0 200]{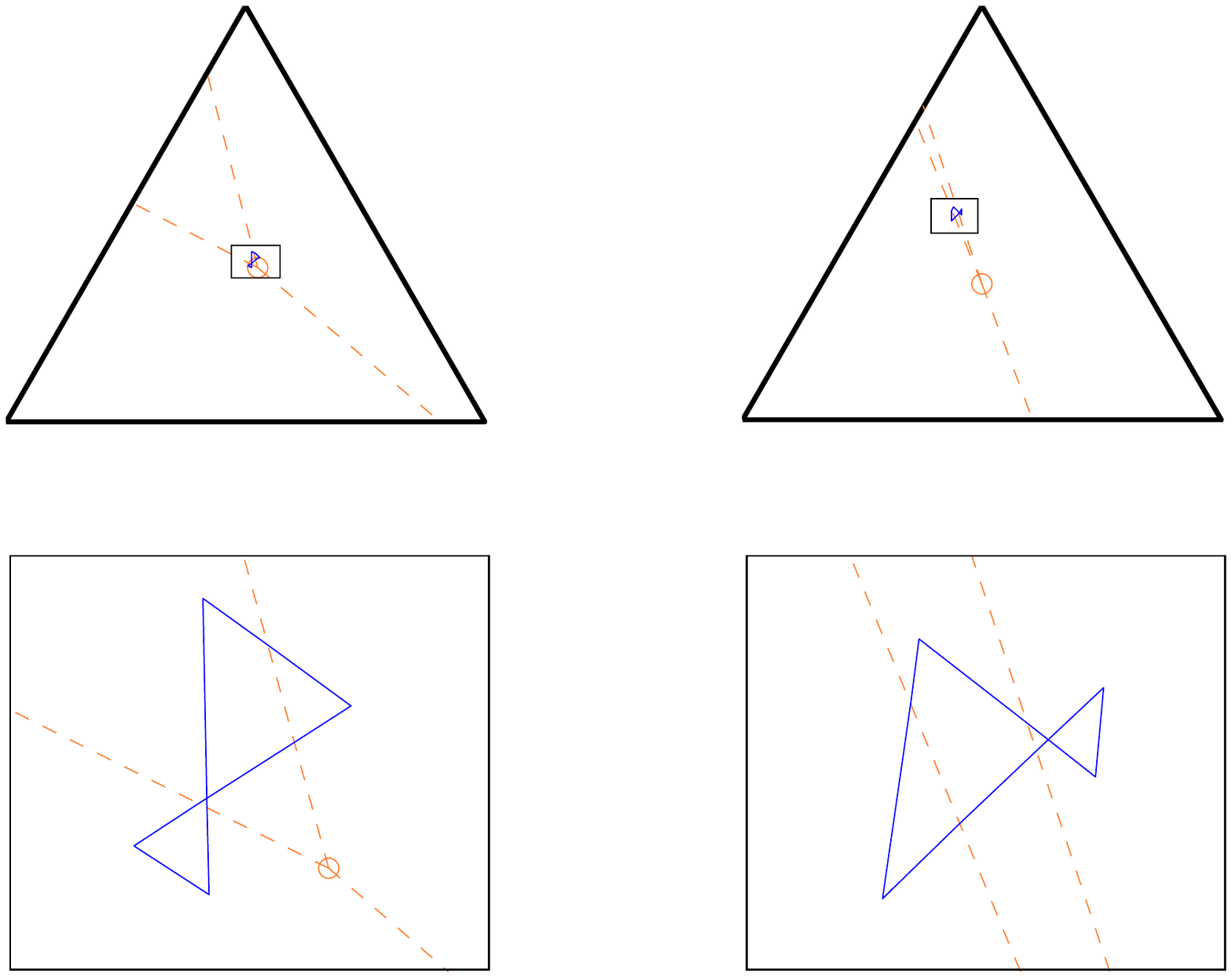}
\caption{Periodic orbit whose projections to $\Siga$ (left) and $\Sigb$ (right) are contained 
in convex cones with apexes $\nasha$ and $\nashb$ respectively. 
The dashed lines indicate the indifference lines of the players. 
Their intersections are the projections of the Nash equilibrium, $\nasha$ and $\nashb$.
For better visibility, the bottom row shows a zoomed version of the periodic orbit.}
\label{fig:per_orbit}
\end{figure}

\begin{figure}[t]
\centering
\includegraphics[width = \textwidth, trim = 50 270 50 250]{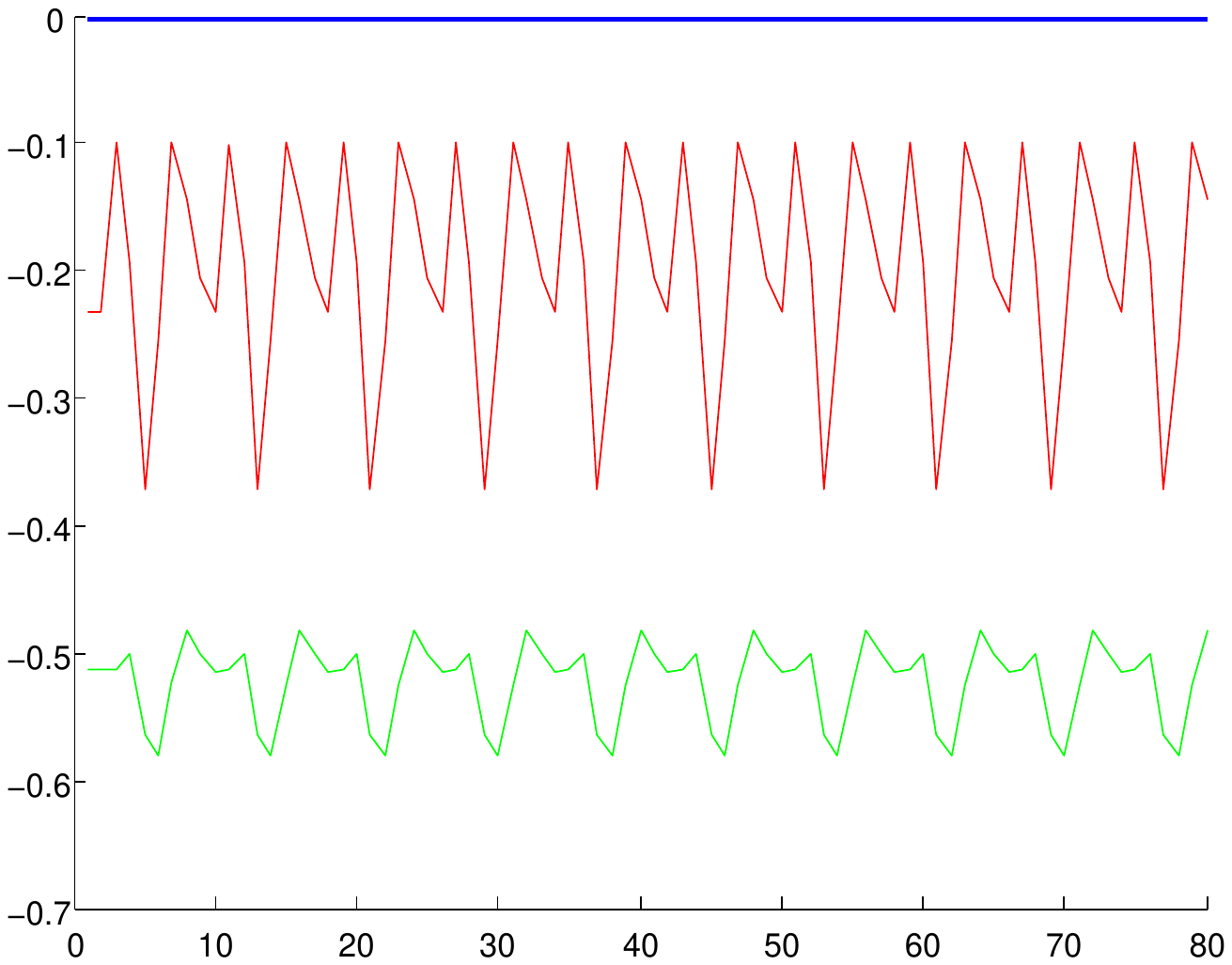}
\caption{Payoff along 10 periods of the periodic orbit contained in $P_B^- \times P_A^-$. 
Player A's payoff oscillates around $-0.5$, player B's payoff around $-0.25$. 
Nash equilibrium payoff is zero to both players.}
\label{fig:per_orbit_po}
\end{figure}
\end{example}

This example has been obtained through numerical experimentation. 
The difficulty in finding an example of a periodic orbit with the 
key property of lying in a convex cone with apex at the unique, completely mixed 
Nash equilibrium seems to suggest that such examples are 
relatively rare. For most games with unique, completely mixed Nash equilibrium, 
payoff along typical FP orbits either Pareto dominates
Nash equilibrium payoff or at least improves upon it along parts of the orbit. 
We formulate the following two conjectures.

\begin{conjecture}
Bimatrix games with unique, completely mixed Nash equilibrium, 
where Nash equilibrium Pareto dominates typical FP orbits are rare.
To be precise, within the space of $n\times n$ games with entries in $[0,1]$, 
those where typical FP orbits are Pareto 
dominated by Nash equilibrium form a set with at most Lebesgue measure $0.01$.
\end{conjecture}

\begin{conjecture} For bimatrix games with unique, completely mixed Nash equilibrium 
and certain transition combinatorics (see~\cite{Ostrovski2010}), 
Nash equilibrium does not Pareto dominate typical FP orbits. 
In particular, this is the case if $\BRa(e_j) \neq \BRa(e_{j'})$ for all $j \neq j'$ and 
$\BRb(e_i) \neq \BRb(e_{i'})$ for all $i \neq i'$.
\end{conjecture}

Indeed, we could strengthen the above conjecture to the following statement.

\begin{conjecture}
 For \q{most} bimatrix games with unique, completely mixed Nash equilibrium, 
 typical FP orbits dominate Nash equilibrium in terms of average payoff. 
 In particular, this is the case under certain assumptions on the transition combinatorics
 of the game; for instance, if each pure strategy invokes a distinct pure best response 
 (as in the previous conjecture). 
\end{conjecture}

\section{Concluding remarks on FP performance}
\label{sec:conclusion_fp_performance}

Conceptually, the overall observation is that playing Nash equilibrium might not be an advantage
over playing according to some learning algorithm (such as FP) in a wide range of games, 
in particular in many common examples of games occurring in the literature. Even in cases
where FP does not dominate Nash equilibrium at all times, 
it might still be preferable in terms of time-averaged payoff. 
In contrast, the previous section shows that there 
are examples in which Nash equilibrium indeed Pareto dominates FP, 
but the restrictive nature of the example suggests that this situation is quite rare. 

Conversely, the discussion also shows that certain notions of game equivalence 
(for instance, linear equivalence, or the weaker best and better response equivalences, 
see~\cite{Morris2004,Moulin1978}), which are popular
in the literature on learning dynamics, are not meaningful in an economic context as 
they do not preserve essential 
features of the payoff structure of games, even though they preserve 
Nash equilibria (and other notions of equilibrium) 
and conditional preferences of the players. 
While some dynamics (in particular, FP dynamics or its autonomous 
version, the best response dynamics \cite{Gilboa1991,Matsui1992,Hofbauer1995}) 
are invariant under all of
these equivalence relations, the actual payoffs 
along their orbits and the payoff 
comparison of different orbits can strongly 
depend on the chosen representative bimatrix, 
as becomes apparent from Theorem~\ref{thm:exist_le}.
This is to some extent analogous to the situation in the 
classical example of the \q{prisoner's dilemma} given by the bimatrix
\[
A = \begin{pmatrix} 3 & 0  \\ 5 & 1 \end{pmatrix}, 
\quad B = \begin{pmatrix} 3 & 5  \\ 0 & 1 \end{pmatrix}.
\]
Under linear equivalence, this corresponds to the bimatrix game
\[
\tilde A = \begin{pmatrix} 0 & 0  \\ 2 & 1 \end{pmatrix}, 
\quad \tilde B = \begin{pmatrix} 0 & 2  \\ 0 & 1 \end{pmatrix},
\]
which shares all essential features such as equilibria, 
best response structures, etc with the prisoner's dilemma. 
Both games are dynamically identical, with all FP
orbits converging along straight lines to the unique pure Nash equilibrium 
$(2,2)$. However, the second game does not 
constitute a prisoner's dilemma in the classical sense: 
whereas in the prisoner's dilemma the Nash equilibrium is Pareto 
dominated by the (dynamically irrelevant) strategy profile $(1,1)$, 
in the second game this is not the case and no \q{dilemma} occurs.

Theorem~\ref{thm:exist_le} can be interpreted in a similar vain: 
linear equivalence turns out to be sufficiently coarse,
so that by changing the representative bimatrix inside an 
equivalence class, one can create certain regions in $\Sigma$ 
in which payoff is arbitrarily high in comparison to the payoff 
at the unique Nash equilibrium. Since FP orbits
remain unchanged, this can be done in such a way that a given 
periodic orbit lies completely or predominantly 
in these desired \q{high payoff portions} of $\Sigma$. On the other hand,
it can be seen from the proof that the conditions for this 
to happen are not at all exceptional. Consequently, it could be argued that in many games of interest the
assumption that Nash equilibrium play is the most desirable outcome might not hold and a more dynamic view of
\q{optimal play} might be reasonable.

\bibliographystyle{abbrv}
\bibliography{../BibTex/library}

\end{document}